\documentclass{llncs}

\usepackage{graphicx,amsfonts,amsmath,amssymb,url,subfig,hyperref}

\usepackage{xcolor}
\usepackage[mathlines]{lineno}
\usepackage{array}

\newcolumntype{C}[1]{>{\centering\arraybackslash}m{#1}}

\allowdisplaybreaks

\newcommand\blfootnote[1]{%
  \begingroup
  \renewcommand\thefootnote{}\footnote{#1}%
  \addtocounter{footnote}{-1}%
  \endgroup}

\DeclareMathOperator{\conv}{\mathit{conv}}

\title{Rectilinear convex hull of points in 3D}

\titlerunning{Rectilinear convex hull of points in 3D}

\author{Pablo P\'erez-Lantero\inst{1}\thanks{Partially supported by projects CONICYT FONDECYT/Regular 1160543 (Chile),
    DICYT 041933PL Vicerrector\'ia de Investigaci\'on, Desarrollo e Innovaci\'on USACH (Chile), and Programa Regional STICAMSUD 19-STIC-02.}
    \and
    Carlos Seara\inst{2}\thanks{Research supported by projects MTM2015-63791-R MINECO/FEDER and Gen. Cat. DGR 2017SGR1640.}
    \and
    Jorge Urrutia\inst{3}\thanks{Research supported in part by SEP-CONACYT of Mexico, Proyecto 80268, and by PAPIIT IN102117 Programa de Apoyo a la Investigaci\'on e Innovaci\'on Tecnol\'ogica, UNAM.}
}

\authorrunning{P\'erez-Lantero et al.}

\institute{Departamento de Matem\'atica y Ciencia de la Computaci\'on, USACH, Chile, \email{pablo.perez.l@usach.cl}
  \and
  Departament de Matem\`{a}tiques, Universitat Polit\`{e}cnica de Catalunya, Spain, \email{carlos.seara@upc.edu}
\and
  Instituto de Matem\'aticas, Universidad Nacional Aut\'onoma de M\'exico, Mexico, \email{urrutia@matem.unam.mx}
}

\begin{document}
  \maketitle

\blfootnote{
 	\begin{minipage}[l]{0.25\textwidth}\includegraphics[trim=10cm 6cm 10cm 5cm,clip,scale=0.15]{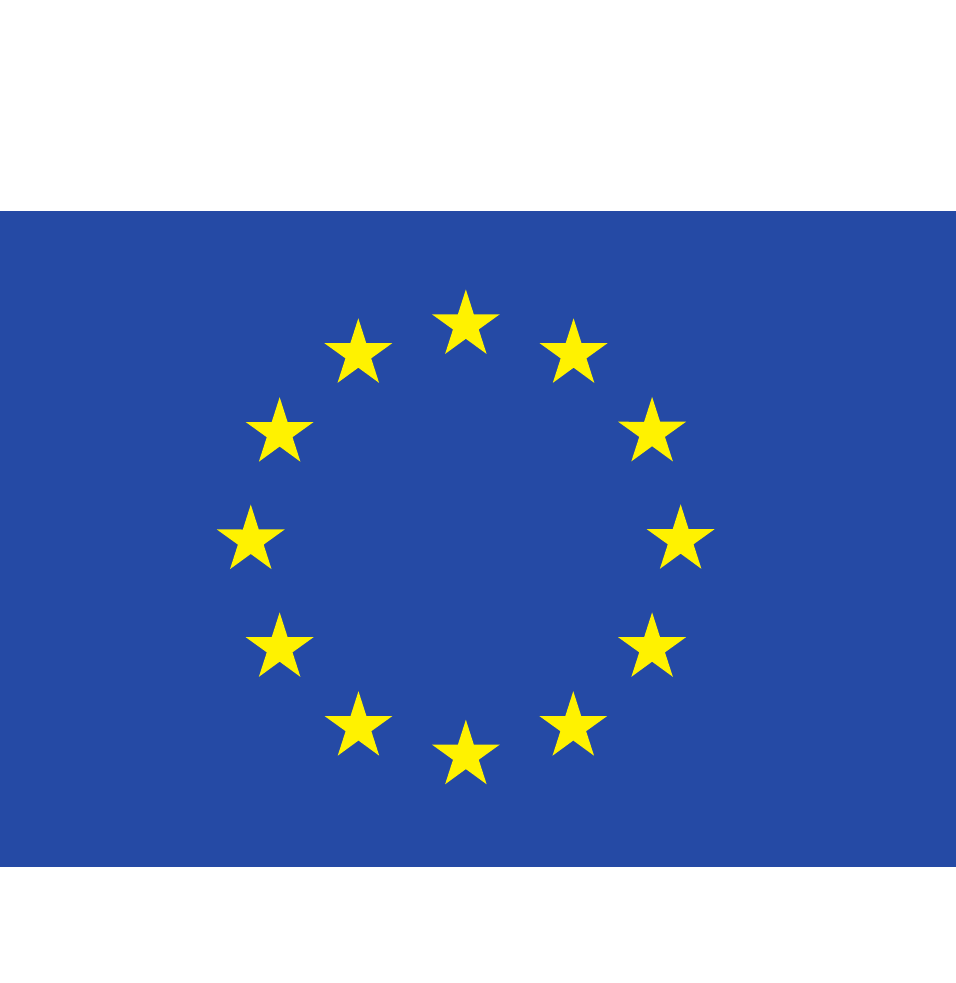}
 	\end{minipage}
 	\hspace{-1.5cm}
 	\begin{minipage}[l][1cm]{0.90\textwidth}This work has received funding from the European Union's Horizon 2020\\
       	research and innovation programme under the Marie Sk\l{}odowska-Curie\\ grant agreement No.\ 734922.
      \end{minipage}
}

\vspace{-0.9cm}

%\linenumbers

%\newtheorem{theorem}{Theorem}
%\newtheorem{claim}{Claim}
%\newtheorem{proposition}{Proposition}
%\newtheorem{corollary}{Corollary}
%\newtheorem{definition}{Definition}
%\newtheorem{problem}{Problem}
\newtheorem{observation}{Observation}
\newtheorem{invariant}{Invariant}
%\newtheorem{lemma}{Lemma}
%\newtheorem{fact}{Fact}
%\newtheorem{remark}{Remark}
%\newtheorem{conjecture}{Conjecture}
%\newtheorem{open}{Open problem}
%\def\proofname{\bf Proof.}

%\newcommand\blfootnote[1]{%
% \begingroup
% \renewcommand\thefootnote{}\footnote{#1}%
% \addtocounter{footnote}{-1}%
% \endgroup}

%\date{}

\begin{abstract}
Let $P$ be a set of $n$ points in $\mathbb{R}^3$ in general position, and let $RCH(P)$ be the rectilinear convex hull of $P$. In this paper we obtain an optimal $O(n\log n)$-time and $O(n)$-space algorithm to compute $RCH(P)$. We also obtain an efficient $O(n\log^2 n)$-time and $O(n\log n)$-space algorithm to compute and maintain the set of vertices of the rectilinear convex hull of $P$ as we rotate $\mathbb R^3$ around the $z$-axis. Finally we study some properties of the rectilinear convex hulls of point sets in $\mathbb{R}^3$.
\end{abstract}

\section{Introduction}\label{sec1}

Let $P$ be a set of $n$ points in the plane. An \emph{open quadrant} in the plane is the intersection of two open half-planes whose supporting lines are parallel to the $x$ and $y$-axes. An open quadrant is called $P$-free if it contains no points of $P$. The rectilinear convex hull of $P$ is the set
\[
  RCH(P)=\mathbb{R}^{2} \setminus\bigcup_{W(P) \in \mathcal{W}}W(P),
\]
where $\mathcal{W}$ denotes the set of all $P$-free open quadrants; see Figure~\ref{fig:rcht}, left.

The rectilinear convex hull of point sets has been studied mostly in the plane; e.g., see Ottmann et al.~\cite{ottmann_1984}, Alegr\'{i}a et al.~\cite{alegria_2014}, and Bae et al.~\cite{bae_2009}.

An open $\theta$-quadrant is the intersection of two open half-planes whose supporting lines are orthogonal, one of which when rotated clockwise by $\theta$ degrees becomes horizontal.

We define the $\theta$-rectilinear convex hull $RCH_{\theta}(P)$ of a point set $P$ as the set
\[
  RCH_{\theta}(P) ~=~ \mathbb{R}^{2} \setminus\bigcup_{W(P) \in \mathcal{W}_{\theta}}W(P),
\]
where $\mathcal{W}_{\theta}$ denotes the set of all $P$-free \emph{open $\theta$-quadrants}.

Note that $RCH_{\theta}(P)$ changes as $\theta$ changes. In fact, as $\theta$ changes from $0$ to $\frac{\pi}{2}$ there are $O(n)$ combinatorially different rectilinear convex hulls; see~\cite{alegria_2014,bae_2009}. Figure~\ref{fig:rcht} right shows an example of a $\theta$-rectilinear convex hull which happens to be disconnected.

\begin{figure}[ht]
  \begin{center}
  \includegraphics[width=0.8 \textwidth]{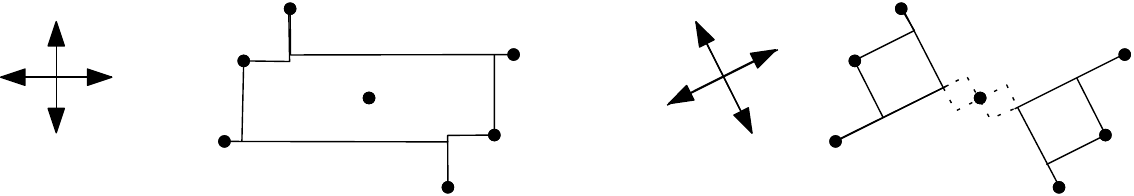}
  \end{center}
  \caption{Left: $RCH(P)$. Right: $RCH_{\pi/6}(P)$ of the same point set.}\label{fig:rcht}
\end{figure}

An open octant in $\mathbb{R}^3$ is the intersection of the three half-spaces, one perpendicular to the $x$-axis, one perpendicular to the $y$-axis, and another one perpendicular to the $z$-axis. As for the planar case, an octant is called $P$-free if it contains no elements of $P$. The rectilinear convex hull of a set of points in  $\mathbb{R}^3$ is defined as
\[
  RCH^3(P)=\mathbb{R}^{3} \setminus\bigcup_{W(P) \in \mathcal{W}}W(P),
\]
where $W$ denotes the set of all $P$-free open octants. In fact, in this paper and as an abuse of language, by $RCH^3(P)$ we will also denote the boundary of $RCH^3(P)$, and analogously for the similar definitions above. Thus, the rectilinear convex layers of a point set in $\mathbb{R}^3$ are defined recursively, as follows: calculate $RCH^3(P)$, and remove the elements of $P$ in $RCH^3(P)$.

\subsubsection{Results.}

In this paper we consider the rectilinear convex hull $RCH^3(P)$ of point sets in $\mathbb{R}^3$. We obtain an $O(n\log n)$ time and $O(n)$ space algorithm to calculate $RCH^3(P)$. We also give an $O(n\log^2 n)$ time and $O(n\log n)$ space algorithm to maintain the set of vertices of $RCH^3(P)$ as we rotate $\mathbb R^3$ around the $z$-axis. We present some results on the combinatorics of rectilinear convex hulls in $\mathbb R^3$ which are related to our algorithmic results, and interesting in their own right. In particular, we show that the rectilinear convex hull of a point set can change a quadratic number of times while its vertex set remains unchanged. Finally we present some open problems.

To avoid cumbersome terminology, from now on we will refer to $RCH^3(P)$ simply as $RCH(P)$.

\subsection{Previous work}

The study of the rectilinear convex hull of point sets in the plane is closely related to that of finding the set of maximal points of point sets under vector dominance. This problem was introduced by Kung et al.~\cite{kung_1975}, see also~\cite{preparata_1985}. They obtained an optimal $O(n\log n)$ time and $O(n)$ space algorithm to solve this problem in the plane and in the three-dimensional space. They also found algorithms to solve the maxima problem in higher dimensions
whose time complexity is $O(n\log^{d-2} n)$ for dimensions $d\ge 3$; however, it is not known whether their algorithm is optimal. Buchsbaum and Goodrich~\cite{buchsbaum04} obtained an algorithm to solve the three-dimensional layers of maxima problem in $O(n\log n)$ time and $O(n\log n/\log\log n)$ space in $\mathbb R^3$. Their algorithm is time optimal.

The rectilinear convex hull of a point set in the plane was first studied by Ottmann et al.~\cite{ottmann_1984} where they obtain an optimal $O(n\log n)$ time algorithm to calculate them. The reader may also consult the monograph Restricted-Orientation Geometry~\cite{fink_2004}, where they study topics related to the problems we study here. Other variants of the problems studied here were also treated~\cite{fink_2004,franek_2009}.

The rectilinear convex layers of a point set $P$ in Euclidean spaces are defined recursively, as follows: calculate the rectilinear convex hull of $P$, remove its elements from $P$, and proceed recursively until $P$ becomes empty. The rectilinear convex layers of point sets were first studied in Pel\'{a}ez et al.~\cite{pelaez_2013}, where an optimal $O(n\log n)$ time algorithm to calculate them was obtained.

Variants of the $RCH(P)$ that were considered in the plane are: computing and maintaining $RCH_{\theta}(P)$ when we rotate the coordinate axes around the origin, or determining the angle of rotation $\theta$ such that the area of $RCH_{\theta}(P)$ is minimized or maximized~\cite{alegria_2014,bae_2009}.

A point set is a \emph{rectilinear convex} set if all of its elements lie on the boundary of their rectilinear convex hull.
Erd\H{o}s-Szekeres type problems for finding rectilinear convex subsets of point sets were studied by Gonz\'alez-Aguilar et al.~\cite{GOPRSTU}. They obtained algorithms to find the largest rectilinear convex subset of a point set, as well as finding their largest rectilinear convex hole, that is, subsets of points of $P$ such that their rectilinear convex hull contains no element of $P$ in the interior.

\subsection{Notation and preliminaries}\label{sec1.1}

For a point $p\in\mathbb{R}^3$, we refer to $x_p$, $y_p$, and $z_p$ as the $x$-, $y$-, and $z$-coordinates of $p$, respectively. A point $p$ satisfies a sign pattern; e.g., $(+,-,+)$ if $x_p\geq 0$, $y_p\leq 0$, and $z_p\geq 0$. There are eight possible sign patterns that a point can satisfy, namely: $(+,+,+)$, $(-,+,+)$, $(-,-,+)$, $(+,-,+)$, $(+,+,-)$, $(-,+,-)$, $(-,-,-)$, and $(+,-,-)$. The first sign pattern corresponds to all of the points in the first \emph{octant} of $\mathbb{R}^3$. Similarly, we will say that the points satisfying the second pattern, $(-,+,+)$, correspond to points in the second octant, \ldots, and those satisfying $(+,-,-)$ correspond to the eighth octant of $\mathbb{R}^3$. The \emph{open} octants are defined in a similar way, except that we require strict inequalities; e.g., the first open octant corresponds to points for which $x_p>0$, $y_p>0$, and $z_p >0$. Given a point $p\in\mathbb{R}^3$, the octants with respect to $p$ are the octants induced by a translation of the origin to $p$.

In addition, each of the eight sign patterns defines a partial order on the elements of $\mathbb{R}^3$ as follows: consider two points $p,q\in\mathbb{R}^3$. We say that $p$ is dominates $q$ according to the sign pattern $(+,+,+)$ if
\[(x_q\le x_p) \wedge (y_q\le y_p) \wedge (z_q\le z_p).\]

We refer to this as $q\preceq_1 p$. In a similar way, we define the domination relations with respect to the other seven sign patterns, which we denote as $p\preceq_2 q$, $p\preceq_3 q$, $p\preceq_4 q$, $p\preceq_5 q$, $p\preceq_6 q$, $p\preceq_7 q$, and $p\preceq_8 q$. For example, the dominance relation with respect to the second sign pattern is:
\[q\preceq_2 p \; \Longleftrightarrow \; (x_q\ge x_p) \wedge (y_q\le y_p) \wedge (z_q\le z_p).\]
These dominance relations define partial orders in $P$, and they can be extended to any dimension $d>3$ in a straightforward way.

\begin{definition}
A point $p\in P$ is a \emph{maximal element} of $P$ with respect to a partial order if there is no $q\in P$ that dominates $p$.
\end{definition}

The partial order $\preceq_1$ defined by $(+,+,+)$ is usually known as \emph{vector dominance}~\cite{kung_1975}.

\section{Computing $RCH(P)$}

In this section we show how to calculate the rectilinear convex hull of point sets in $\mathbb R^3$. The problem of finding the maximal elements of a point set in $\mathbb{R}^2$ and $\mathbb{R}^3$ with respect to vector dominance is called the maxima problem. The next result is well known.

\begin{theorem}{\rm\cite{kung_1975}}\label{teo1}
The maxima problem in $\mathbb{R}^d$, $d=2,3$, can be solved in optimal $O(n\log n)$ time and $O(n)$ space.
\end{theorem}

In fact, when solving the maxima problem, we obtain the set of faces and vertices of the \emph{orthogonal
polyhedron}, i.e., the faces meet at right angles and edges are parallel to the axes; call it $\mathcal P^1$ as shown in Figure~\ref{fig:overlap1} left. In the right part of Figure~\ref{fig:overlap1} we show the top view of $\mathcal P^1$. Observe that if we intersect a horizontal plane $\mathcal{Q}_c$ with equation $z=c$ with $\mathcal P^1$, we obtain an orthogonal polygon $\mathcal{C}^1_c$ that will change as we move $\mathcal{Q}_c$ up or down; i.e., as we increase or decrease $c$.

\begin{figure}[ht]
  \begin{center}
  \includegraphics[width=0.7\textwidth]{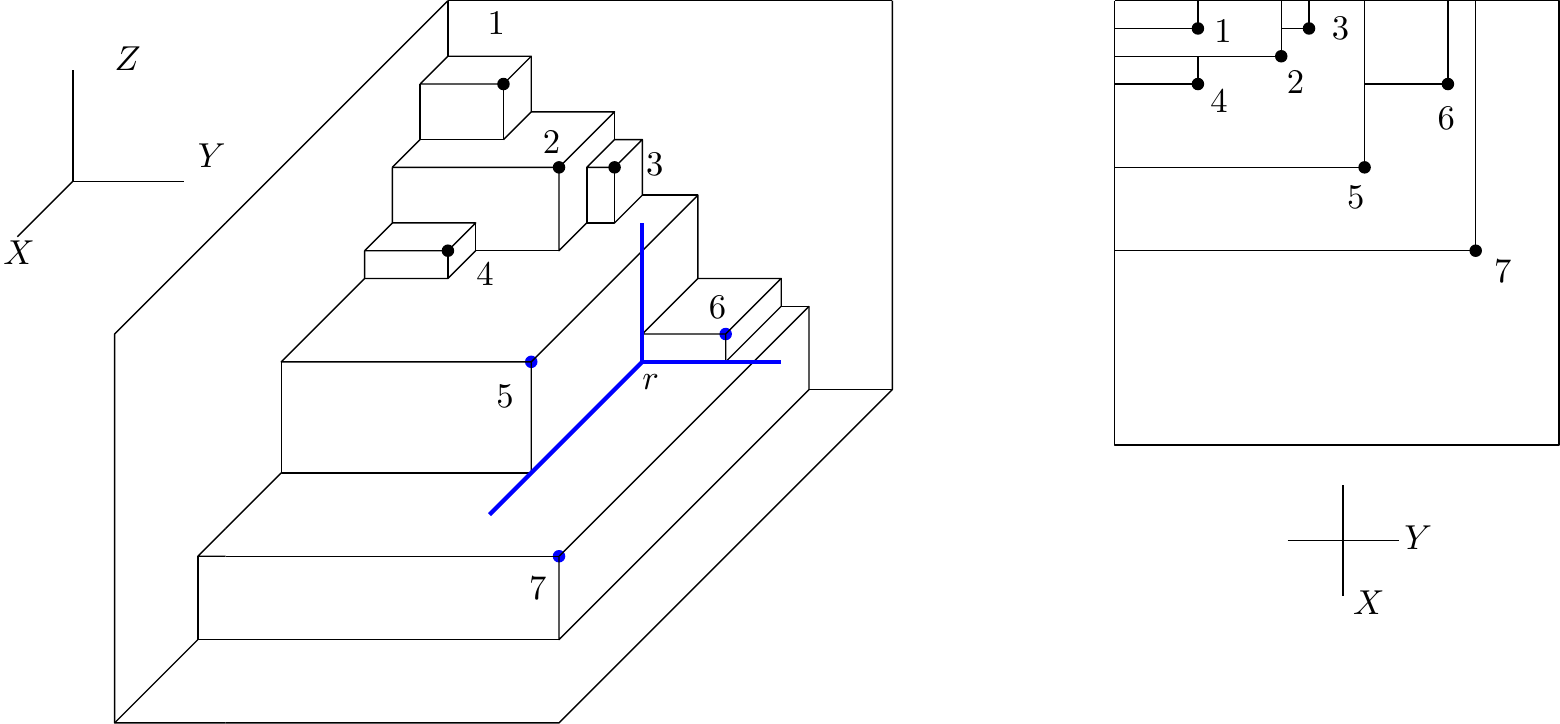}
  \end{center}
  \caption{Left: Maxima points in the first octant with the topmost antenna in $p_1$, and an extremal first open octant in blue. Right: The projection on the $XY$ plane of the the first octant maxima point set.}\label{fig:overlap1}
\end{figure}

In Figure~\ref{changeC} we show how $\mathcal{C}^1_c$ changes as we scan it from top to bottom starting at the top vertex of $\mathcal{P}^1$. We show the intersection of $\mathcal{P}^1$ with $\mathcal{Q}_c$ as $\mathcal{Q}_c$ sweeps through the first four top vertices of $\mathcal{P}^1$. It is easy to see that when $\mathcal{Q}_c$ moves from one vertex of $\mathcal{P}^1$ to the next, $\mathcal{C}^1_c$ changes in the following way: a new vertex appears, which is the vertex of an \emph{elbow} from which two rays emanate, one horizontal and one vertical, that extend until they hit $\mathcal{C}^1_c$ or go to infinity; see Figure~\ref{changeC}.

\begin{figure}[ht]
  \begin{center}
  \includegraphics[width=1.0\textwidth]{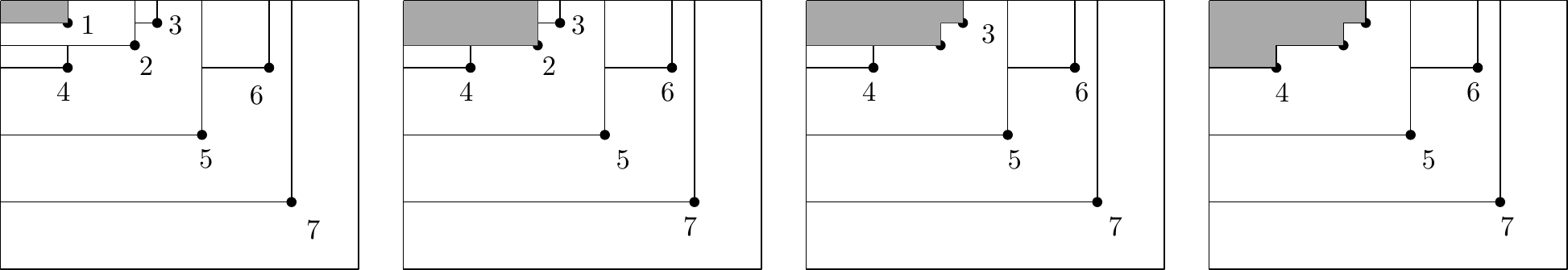}
  \end{center}
  \caption{How $\mathcal{C}^1_c$ changes as $\mathcal{Q}_c$ moves down from point $1$ to point $4$.}\label{changeC}
\end{figure}

Analogously, we can define $\mathcal P^i$, $\mathcal C^i_c$, $i=2,\ldots,8$. Since $RCH(P)=\bigcap_{i=1,\ldots,8} \mathcal{P}^i$, we have the following.

\begin{theorem}
For each constant $c$, the intersection of $RCH(P)$ with $\mathcal{Q}_c$ is the intersection of the orthogonal polygons $\mathcal{C}^i_c$, $i=1,\ldots,8$.
\end{theorem}

To compute $RCH(P)$ we will sweep a plane $\mathcal{Q}_c$ from top to bottom, stopping at each point of $P$ on $RCH(P)$. Each time we stop, we need to update $RCH(P)=\bigcap_{i=1,\ldots,8} \mathcal{P}^i$. We claim that we can do this in $O(\log n)$ time. To prove this, observe that when we move $\mathcal{Q}_c$ from a point $p$ of $P$ to the next point, say $q$, the only curves $\mathcal{C}^i_c$, $i=2,\ldots,8$ that change are those containing $q$, and these can be recomputed in $O(\log n)$ time.

As a consequence of the discussion above we have the following result.

\begin{theorem}\label{teo2}
Given a set $P$ of $n$ points in 3D, the rectilinear convex hull of $P$, $RCH(P)$, can be computed in optimal $O(n\log n)$ time and $O(n)$ space.
\end{theorem}

\section{Maintaining $RCH_{\theta}(P)$}\label{sec3}

In the plane, the problem of maintaining $RCH_{\theta}(P)$ as $\theta$ changes from $0$ to $2 \pi $ has been studied~\cite{alegria_2014,bae_2009}. In this section we will study this problem in $\mathbb{R}^{3}$ restricted to rotations of $\mathbb{R}^{3}$ around the $z$-axis. Thus, in the rest of this section we will use octants defined as intersections of three mutually orthogonal semi-spaces whose supporting planes are orthogonal to three mutually orthogonal lines through the origin, one of which is the $z$-axis. Thus, two of these three lines lie on the $XY$-plane, and correspond to rotations of the $x$- and $y$-axis by an angle $\theta$ in the clockwise direction. We call such octants $\theta$-octants, and the corresponding rectilinear convex hulls generated $RCH_{\theta}(P)$. In the rest of this section we will assume that elements of $P$ are labeled $\{p_1,\ldots,p_n\}$ from top to bottom according to their $z$-coordinate.

For every  $p\in\mathbb{R}^{3}$ there are eight $\theta$-octants having $p$ as their apex; we will call them $p^{\theta}$-octants. Exactly four $p^{\theta}$-octants contain points in $\mathbb{R}^{3}$ above the horizontal plane $\lambda_p$ through $p$, and the other four have points below $\lambda_p$. We call the first four \emph{up $p^{\theta}$-octants}, and the other \emph{down $p^{\theta}$-octants}. Note that if an up $p^{\theta}$-octant is no-$P$-free, it contains elements of $P$ above $\lambda_p$, and that non-$P$-free down $p^{\theta}$-octants contain points in $P$ below $\lambda_p$.

Observe that a point $p\in P$ is a vertex of $RCH_{\theta}(P)$ if there is a $P$-free $p^{\theta}$-octant. In this case we will say that $p$ is a $\theta$-active point, otherwise $p$ is $\theta$-inactive. Furthermore, if there is a $P$-free up $p^{\theta}$-octant, we call $p$ an up $\theta$-active vertex. We define down $\theta$-active vertices in a similar way.

We first analyze the set of angles for which points in $P$ are up $\theta$-active. Let $p_i$ be a point of $P$, and consider the orthogonal projection onto $\lambda_{p_i}$ of the points $p_1,\ldots,p_{i-1}$ (the points in $P$ above $\lambda_{p_i}$), and let $P'_i$ be the point set thus obtained; see Figure~\ref{fig:overlap4} left. If for some $\theta$, $p_i$ is up $\theta$-active, then there is a wedge of angular size $\theta_{p_i}$ at least $\frac{\pi}{2}$ on $\lambda_{p_i}$ whose apex is $p_i$, that is $P'_i$-free; see Figure~\ref{fig:overlap4} right. Clearly, no more than three such disjoint wedges can exist. In a similar way, we can prove that $p_i$ is down active in at most three angular intervals. Thus, the following result, equivalent to a result in Avis et al.~\cite{theta-maxima_1999} for points on the plane, follows.

\begin{theorem} \label{threeintervals}
The set of angles $\theta$ for which a point $p_i\in P$ is active consists of at most six disjoint intervals in the set of directions $[0,2\pi]$.
\end{theorem}

\begin{figure}[ht]
  \begin{center}
  \includegraphics[width=0.75\textwidth]{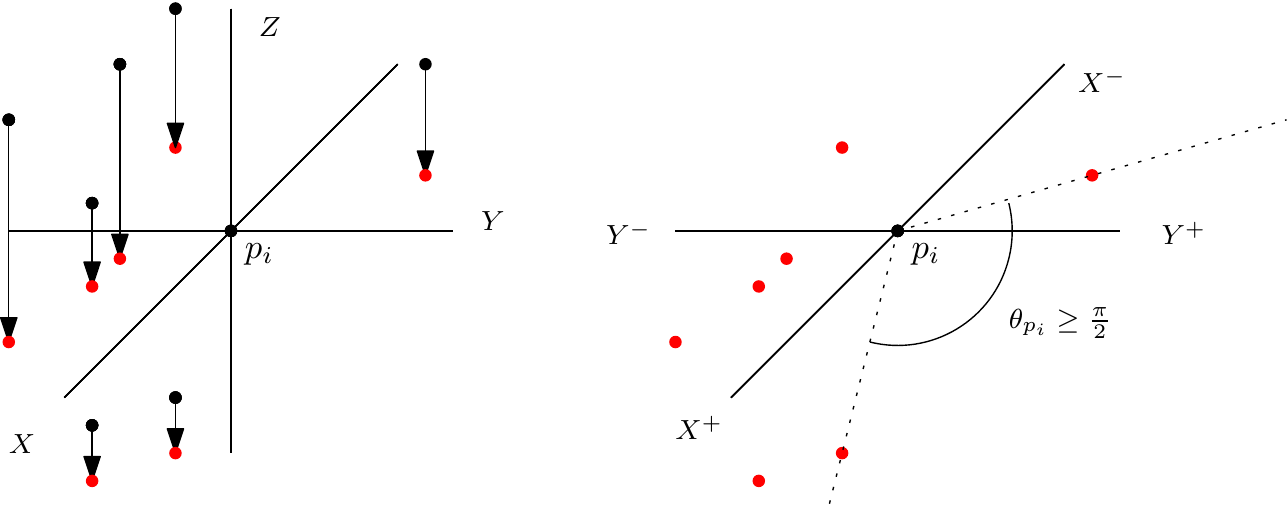}
  \end{center}
  \caption{Checking whether the point $p_i$ is maximal with respect to the first octant for some angular interval: project the set of points $\{p_1,p_2,\dots,p_{i-1}\}$ on the plane parallel to the $XY$ plane passing through $p_i$, and obtain the projected points (red points in figure).}\label{fig:overlap4}
\end{figure}

Finding the angle intervals at which $p_i$ is up-active is now reduced to finding, if they exist, $P'_i$-free wedges in $\lambda_{p_i}$ whose apex is $p_i$ of angular size at least $\frac{\pi}{2}$. We solve this as follows: note first that if one such wedge exists, it has to contain at least one of the four rays emanating from $p_i$ parallel to the $x$- or $y$-axis. Let those rays be $X_i^+$, $X_i^-$, $Y_i^+$, and $Y_i^-$; see Figure~\ref{fig:overlap4} right. For $X_i^+$, we will solve the following problem.

Rotate $X_i^+$ clockwise until it hits a point in $P'_i$. Next, rotate $X_i^+$ clockwise until it hits another point in $P'$. Measure the angle $\alpha$ formed by the two rays thus obtained. If $\alpha\geq\frac{\pi}{2}$ then we have found a set of intervals at which $p_i$ is active, else discard $X_i^+$. Proceed in the same way with $X_i^-$, $Y_i^+$, and $Y_i^-$. We will have to repeat this process for all of the points $p_i\in P$ from top to bottom. We now show how to process all the points in $P$ in $O(n\log^2 n)$ time and $O(n\log n)$ space.

The main difficulty in finding the wedges in the above discussion is that as we process the points of $P$ from top to bottom, the number of points in $\lambda_{p_i}$ increases one by one, and thus, we need a dynamic data structure to solve the following problem.

\begin{problem}\label{wedges}
Let $Q=\{q_1,\ldots,q_n\}$ be a set of points in the plane, and let $Q_{i-1}=\{q_1,\ldots,q_{i-1}\}$. For each $i$ we want to solve the following problem: let $r_i$ be the vertical ray that starts at $q_i$ and points up. Find the first point $a_i$ (respectively, $b_i$) of $Q_{i-1}$ that $r_i$ meets when we rotate it in the clockwise (respectively, counter-clockwise) direction around $q_i$.
\end{problem}

One last point before proceeding with our results; instead of projecting the points of $P$ on the planes $\lambda_{p_i}$, we will project them one by one on the $XY$-plane. Everything else remains unchanged. The next observation on binary trees will be useful.

\begin{observation} \label{partition}
Let $\mathcal{T}$ be a balanced binary tree. For each node of $v\in\mathcal{T}$ let $S_v$ be the set of leaves of $\mathcal{T}$ that are descendants of $v$. Let $u$ be a leaf of $\mathcal{T}$, and consider the path $p_u$ that joins $u$ to the root of $\mathcal{T}$, and let $\Delta_u$ be the set of nodes of $\mathcal{T}$ that are direct descendants of a node in $p_u$. Then, the sets $S_v$, $v\in\Delta_u$ induce a partition of the set of leaves of $\mathcal{T}\setminus \{u\}$. Moreover if a node $w$ in $\Delta_u$ is the right (left) descendant of a node in $p_u$, then the elements of $S_w$ lie to the right (respectively, left) of $v$; see Figure~\ref{fig:overlap611}.
\end{observation}

\begin{figure}[ht]
  \begin{center}
  \includegraphics[width=0.8\textwidth]{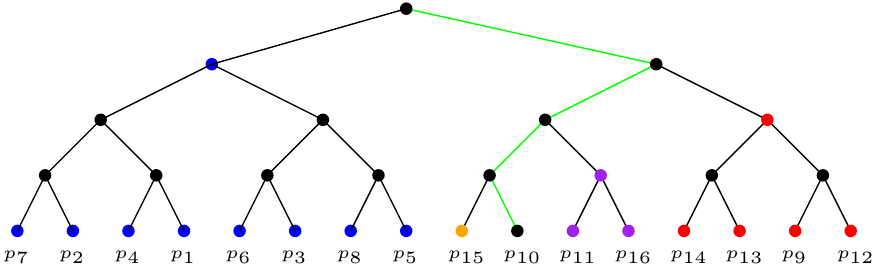}
  \end{center}
  \caption{}\label{fig:overlap611}
\end{figure}

\begin{theorem}\label{teo3}
Problem~\ref{wedges} can be solved in $O(n\log^2 n)$ time and $O(n\log n)$ space.
\end{theorem}

\begin{proof}
Assume without loss of generality that no two points of $Q$ lie on a horizontal line, and let ${\mathit{left}}_i$ and ${\mathit{right}}_i$ be the set of points in $P$ lying to the left (respectively, right) of the vertical line through $q_i$.
If we know the convex hulls of ${\mathit{left}}_i$ and ${\mathit{right}}_i$, then the points we are seeking, $a_i$ and $b_i$,
can be computed by calculating the supporting lines of the convex hull of ${\mathit{left}}_i$ and ${\mathit{right}}_i$ passing through $q_i$. It is well known that we can compute these lines in $O(\log n)$ time.

Furthermore, if ${\mathit{left}}_i$ and ${\mathit{right}}_i$ have been decomposed into $k$ disjoint sets $W_i,\ldots,W_k$ where each $W_i$ is contained in ${\mathit{left}}_i$ or in ${\mathit{right}}_i$, and we have the convex hull $\conv(W_i)$ of all of these point sets, then we can find the supporting lines through $q_i$  for all of them in overall $O(\log|W_1|+\dots+\log|W_k|)=O(\log^2 n)$ time; see Figure~\ref{fig:overlap6_12}. This will now allow us to obtain $a_i$ and $b_i$ in $O(\log n)$ time. This is the main idea that will enable us to design a data structure to solve Problem~\ref{wedges} in $O(n\log^2 n)$ time.

\begin{figure}[ht]
  \begin{center}
  \includegraphics[width=0.8\textwidth]{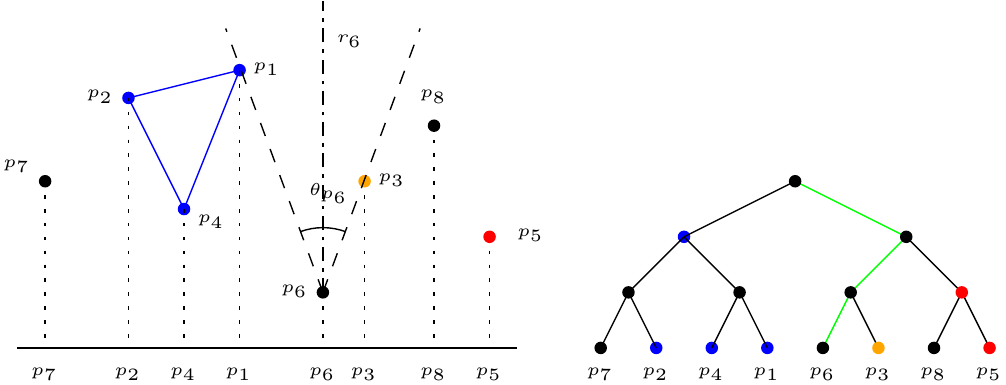}
  \end{center}
  \caption{Processing $p_6$. At this point, $p_7$ and $p_8$ are inactive, thus they are shown in black.}\label{fig:overlap6_12}
\end{figure}

Let $D(Q)$ be a balanced binary tree whose leaves are the elements of $Q$ sorted in order from left to right according to their $x$-coordinate. Note that this order does not necessarily coincide with the labeling $q_1,\ldots,q_n$ of the elements of $Q$. Initially, every leaf of $D(Q)$ is considered \emph{inactive}. For each vertex of $D(Q)$ we maintain the convex hull of $W(q)$, where $W(q)$ is the set of descendant leaves (points of $Q$) that are active.

For each $i$, from $i=1,\ldots,n$ we execute the following algorithm:
\begin{itemize}
\item Consider the nodes of $D(Q)$ that are direct descendants of nodes in the path $p_{q_i}$ connecting $q_i$ to the root of $D(Q)$. By Observation~\ref{partition}, the active descendants of these nodes form a partition $W_i,\ldots,W_k$ of the set of active leaves of $D(Q)\setminus \{q_i\}$, and each of these sets is contained to the left or the right of the vertical line through $q_i$. Moreover we know the convex hulls of each of $W_i,\ldots,W_k$. Thus, we can calculate their supporting lines passing through $q_i$ in overall $O(\log^2 n)$ time.

\item Make $q_i$ active, and update the convex hulls stored at the vertices of the path joining $q_i$ to the root of $D(Q)$. This can be done in logarithmic time per node of the path, and overall $O(\log^2 n)$ time.
\end{itemize}

Observe that the convex polygon associated with the root of $D(Q)$ can be of size $n$. For the vertices in the next level of $D(Q)$, the sum of the sizes of the convex polygons associated to them is $n$, and in general, the sum of the polygons associated to the vertices of $D(Q)$ is $n$. Since the number of levels of $D(Q)$ is $O(\log n)$, the space used is $O(n\log n)$. \qed
\end{proof}

By using the results in Theorem~\ref{teo3} we can calculate the set of intervals at which all of the points in $P$ are up-active. In a similar way we can determine the intervals for which the points in $P$ are down-active; thus we have the following.

\begin{theorem}
The set of intervals at which the points of $P$ are $\theta$-active can be computed in $O(n\log^2 n)$ time and $O(n\log n)$ space.
\end{theorem}

Observe that the set of intervals at which two points of $P$ are active define intervals in the unit circle $C$, where the points on $C$ correspond to angles in $[0,2\pi]$. Thus, if an angle $\theta$ belongs to an interval at which a point of $P$ is active, this point is a vertex of $RCH_{\theta}(P)$. As $\theta$ goes from $0$ to $2\pi$, the vertices of $RCH_{\theta}(P)$ are those for which one of its active angular intervals contains $\theta$.

As a consequence of the above discussion we have the following result.

\begin{theorem}\label{teo4}
Given a set $P$ of $n$ points in 3D, maintaining the elements of $P$ that belong to the boundary of $RCH_{\theta}(P)$ as $\theta\in [0,2\pi]$ can be done in $O(n\log^2 n)$ time and $O(n\log n)$ space.
\end{theorem}

Note that if we store the set of angular intervals at which the points of $P$ are active, then we can, for any angle $\theta$ retrieve the points in $P$ that are $\theta$ active in linear time. In case that we want to compute the $RCH_{\theta}(P)$ for a particular value of $\theta$, all we have to do is to retrieve that $\theta$-active points of $P$ in linear time, and use the algorithm presented in Theorem~\ref{teo2}.

\section{The Combinatorics of Rectilinear Convex Hulls in $\mathbb{R}^3$}

In the previous section, we studied the problem of maintaining the set of points of $P$ that are vertices of $RCH_{\theta}(P)$.
The problem of maintaining $RCH_{\theta}(P)$ does not follow from our previous results. As we shall see, there are examples of point sets $P\subset\mathbb{R}^3$ such that the number of combinatorially different rectilinear convex hulls of $P$ can be at least $\Omega(n^2)$ while the set of vertices of $RCH_{\theta}(P)$ remains unchanged.

To begin, we notice that $RCH_{\theta}(P)$ is not necessarily connected; this is easy to see, as for point sets in the plane this property does not hold; e.g., see Figure~\ref{fig:rcht} right. What is a bit more interesting is that even when $RCH_{\theta}(P)$ is connected and has non-empty interior, it is not necessarily simply connected. In Figure~\ref{torus} we show a rectilinear convex hull of a point set whose rectilinear convex hull is a torus. The elements of $P$ are the vertices of the cubes glued together to obtain the figure. The reader will notice immediately that the points in $P$ are not in general position. A slight perturbation of the elements of $P$, that would bring them to a point set in general position, will also yield a rectilinear convex hull whose interior is a torus. Evidently, we can construct similar examples to that shown in Figure~\ref{torus} in which we obtain oriented surfaces of arbitrarily large genus.

\begin{figure}[ht]
  \begin{center}
  \includegraphics[width=0.6\textwidth]{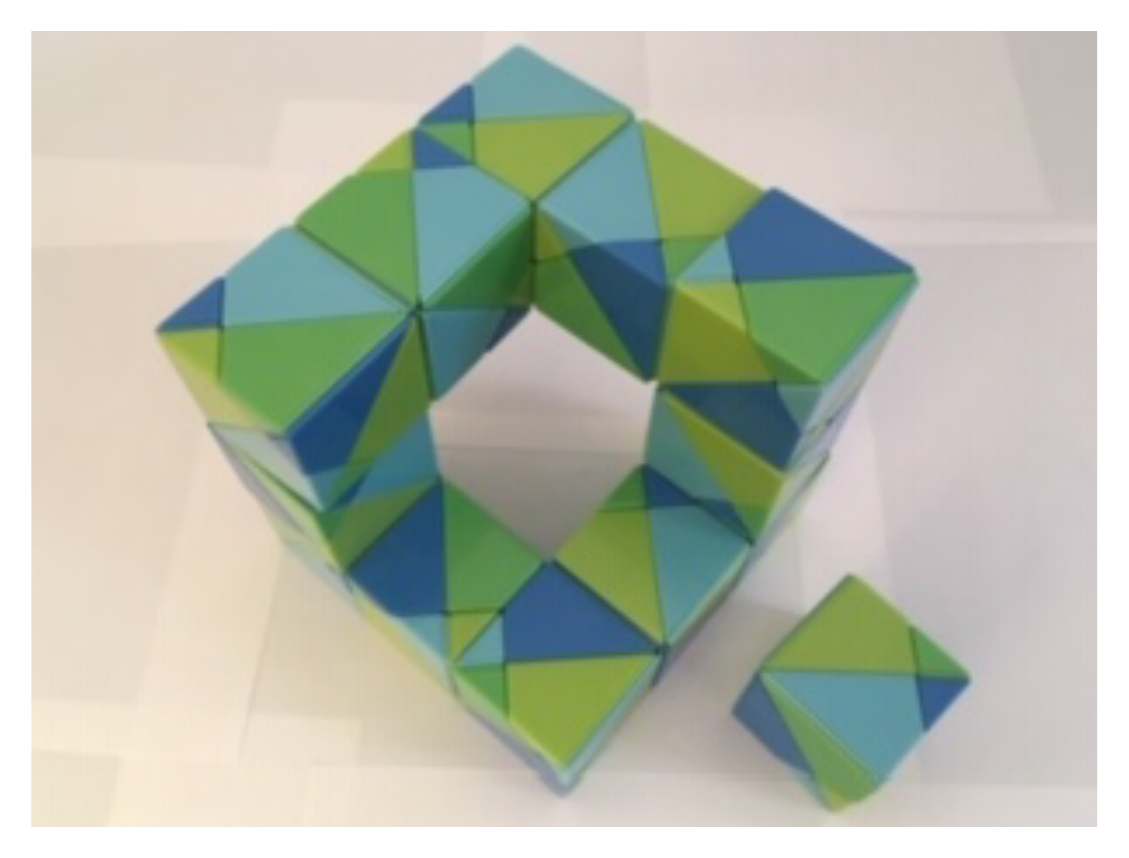}
  \end{center}
  \caption{The rectilinear convex hull of a point set that is not simply connected.}\label{torus}
\end{figure}

We now construct a point set such that for an angular interval $[\alpha,\beta]$ while $\theta\in[\alpha,\beta]$, $RCH_{\theta}(P)$ will maintain the same vertices while it changes a quadratic number of times.

Consider a circular cylinder $\mathcal{C}$ that is perpendicular to the $XY$-plane, and consider a geodesic curve $\mathcal{H}$ on $\mathcal{C}$ that joins two points $p$ and $q$ on $\mathcal{C}$. Choose a set of $n$ points $P'=\{p'_1,\ldots,p'_n\}$ on $\mathcal{H}$ such that their projection on the $XY$-plane is a set of equidistant points on a small interval of the circle in which $\mathcal{C}$ and the $XY$-plane intersect, and such that if $i<j$ the $z$-coordinate of $p_i$ is smaller than the $z$-coordinate of $p_j$; see Figure~\ref{pfreeoctants}.

\begin{figure}[ht]
  \begin{center}
  \includegraphics[width=0.8\textwidth]{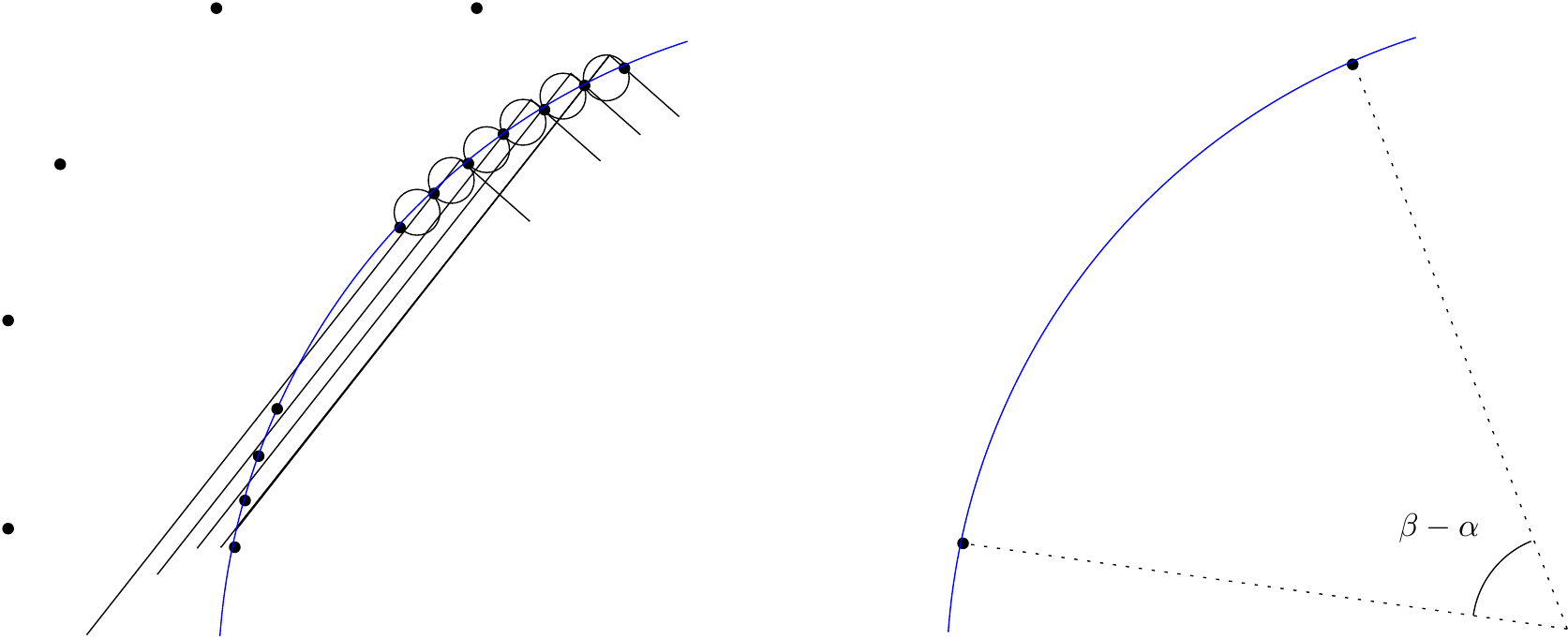}
  \end{center}
  \caption{A configuration of points that illustrates a point set for which its set of vertices remains unchanged while its rectilinear convex hull changes a quadratic number of times.}\label{pfreeoctants}
\end{figure}

Observe that there is a $P'$-free up octant $Q_{1,n}$ whose bottom, left, and back faces contain $p'_1$, $p'_{n-1}$ and $p'_{n}$. Observe now that there is a translation of $Q_1$ up and to the left so as to produce a second octant $Q_{1,n-1}$ whose bottom, left and back faces contain $p'_2$, $p'_{n-1}$, and $p'_{n-2}$. We can iterate this process $\frac{n-1}{2}$ times to obtain a set of  $\frac{n-1}{2}$ $P'$-free octants that have on their bottom, left and back faces $p'_1$, $p'_i$ and $p'_{i+1}$, $i=\lfloor\frac{n-1}{2}\rfloor,\ldots,n-1$; see Figure~\ref{pfreeoctants}.

Rotating $Q_2$ slightly in the clockwise direction and moving it up we can obtain a $P$-free octant $Q_{2,n}$ containing on its bottom, left, and back faces $p'_2$, $p'_n$, and $p'_{n-1}$. We can now repeat the same process as we did with $\{p'_1,\ldots,p'_n\}$, and with $\{p'_2,\ldots,p'_n\}$, starting with $p'_2$ and $Q_{2,n}$ to obtain a new set of $P'$-free extremal octants. Repeat this process with $\{p'_3,\ldots,p'_n\}$,\ldots, $\{p'_{n-3},\ldots,p'_n\}$, to obtain a quadratic number of $P'$-free \emph{extremal} $P'$-octants. Let $\alpha$ and $\beta$ be the angles of rotation of the $XY$-plane such that the $x$-axis becomes parallel to the line at which the back faces of $Q_{1,n}$ and $Q_{n-3,n}$ intersect the $XY$-plane. Observe that all of $p'_1, \ldots , p'_n$ are active points and on the boundary of $RCH_{\theta}$ for all $\theta\in[\alpha,\beta]$. In the meantime, all of the $P$-free octants we obtained above become active during an angular interval contained in $[\alpha,\beta]$. To complete our construction, we add a few points to the set $\{p'_1,\ldots,p'_n\}$. These points are located behind the circular cylinder $\mathcal{C}$, placed appropriately to ensure that the rectilinear convex hull of the point set $P$ thus obtained has non-empty interior, and all of $\{p'_1,\ldots,p'_n\}$ are on its boundary. Thus, the rectilinear convex hull of $P$ changes a quadratic number of times while its vertex set remains unchanged.

\begin{theorem}
There are configurations of points in $\mathbb{R}^3$ such that for an angular interval $[\alpha,\beta]$ while $\theta\in [\alpha,\beta]$, $RCH_{\theta}(P)$ will maintain the same vertices while it changes a quadratic number of times.
\end{theorem}

\section{Final Remarks and Future Lines of Research}\label{sec4}

We have shown how to calculate the rectilinear convex hull of a point set in $O(n\log n)$ time. The rectilinear convex layers of a point set in $\mathbb{R}^3$ are defined in a recursive way as follows: calculate $RCH^3(P)$, and remove the elements of $P$ in $RCH^3(P)$. It is clear that the rectilinear convex layers of $P$ can be computed in a recursive way by removing the rectilinear convex hull of the point set until $P$ becomes empty. If $P$ has $k$ rectilinear convex layers, this can be done in $O(kn\log n)$ time and $O(n)$ space. We conjecture that there exists an algorithm to find the rectilinear convex layers of $P$ in better than $O(kn\log n)$ time and $O(n)$ space.

We proved that the number of times that the set of vertices of the rectilinear convex hull of a point set $P$ changes while $\mathbb R^3$ is rotated around the $z$-axis is linear; however, the rectilinear convex hull of $P$ may change a quadratic number of times. A future line of research is that of obtaining efficient algorithms to maintain the rectilinear convex hull of $P$ in time proportional to the number of times it changes times a logarithmic factor.

Finally, we remark that obtaining the rectilinear convex hull of a point set, when we rotate $\mathbb R^3$ around any line through the origin, can be done trivially in $O(n\log n)$ time, as any such rotation can be achieved as a composition of a rotation around the $z$-axis followed by a rotation around the $x$-axis.

%\newpage

%\section*{Appendix A: Computing the rectilinear convex layers of $P$ in 3D}\label{Asec4}

%\section*{Appendix B: An application to solve a 3D fitting problem}\label{Asec5}

\end{document}